\documentclass[a4paper,UKenglish,cleveref, autoref, thm-restate]{lipics-v2021}

\title{Make Every Word Count: Adaptive Byzantine Agreement with Fewer Words}

\author{Shir Cohen}{Technion, Israel}{shirco@campus.technion.ac.il}{}{Supported by the Adams Fellowship Program of the
Israel Academy of Sciences and Humanities.}
\author{Idit Keidar}{Technion, Israel}{idish@ee.technion.ac.il}{}{}
\author{Alexander Spiegelman}{Aptos, USA}{sasha.spiegelman@gmail.com}{}{}

\authorrunning{S. Cohen, I. Keidar, and A. Spiegelman} 

\Copyright{Shir Cohen, Idit Keidar, and Alexander Spiegelman} 



\category{} 

\relatedversion{} 

\nolinenumbers 

\hideLIPIcs

\EventEditors{Eshcar Hillel, Roberto Palmieri, and Etienne Rivi\`{e}re}
\EventNoEds{3}
\EventLongTitle{26th International Conference on Principles of Distributed Systems (OPODIS 2022)}
\EventShortTitle{OPODIS 2022}
\EventAcronym{OPODIS}
\EventYear{2022}
\EventDate{December 13--15, 2022}
\EventLocation{Brussels, Belgium}
\EventLogo{}
\SeriesVolume{253}
\ArticleNo{13}

\usepackage{amsmath}
\usepackage{xspace}
\usepackage{mathtools}

\usepackage{algorithm}
\usepackage[noend]{algpseudocode}

\newcommand{\lr}[1]{\langle #1 \rangle}

\newcommand{\helpreq}{\mathsf{help\_req}\xspace}
\newcommand{\idk}{\mathsf{idk}\xspace}

\newcommand{\propose}{\mathsf{propose}\xspace}
\newcommand{\commit}{\mathsf{commit}\xspace}
\newcommand{\vote}{\mathsf{vote}\xspace}
\newcommand{\decide}{\mathsf{decide}\xspace}
\newcommand{\finalize}{\mathsf{finalize}\xspace}

\newcommand{\helpre}{\mathsf{help\_req}\xspace}
\newcommand{\help}{\mathsf{help}\xspace}
\newcommand{\fb}{\mathsf{fallback}\xspace}

\DeclarePairedDelimiter{\ceil}{\lceil}{\rceil}

\begin{document}

\maketitle

\begin{abstract}
Byzantine Agreement (BA) is a key component in many distributed systems. While Dolev and Reischuk have proven a long time ago that quadratic communication complexity is necessary for worst-case runs, the question of what can be done in practically common runs with fewer failures remained open.
In this paper we present the first Byzantine Broadcast algorithm with $O(n(f+1))$ communication complexity in a model with resilience of $n=2t+1$, where $0\leq f\leq t$ is the actual number of process failures in a run. And for BA with strong unanimity, we present the first optimal-resilience algorithm that has linear communication complexity in the failure-free case and a quadratic cost otherwise.
\end{abstract}



\section{Introduction}

Byzantine Agreement (BA) is a key component in many distributed systems.
As these systems are being used at larger scales, there is an increased need to find efficient solutions for BA.
Arguably, the most important aspect of an efficient BA solution is its communication costs.
That is, how much information needs to be transferred in the network to solve the BA problem.
Indeed, improving the communication complexity, often measured as word complexity, was the focus of many recent works and deployed systems~\cite{Abraham2018HotStuffTL,Algorand,spiegelman2021search,abraham2019asymptotically,momose2020optimal,cohen2020not}.

In the BA problem, a set of $n$ processes attempt to agree on a \emph{decision} value despite the presence of Byzantine processes.
One of the properties of a BA algorithm is a threshold $t$ on how many Byzantine processes it can withstand.
Namely, the algorithm is correct as long as up to $t$ processes are corrupted in the course of a run.
In this paper we focus on $n=2t+1$ and we assume a trusted setup of a public-key infrastructure (PKI) that enables us to use a threshold signature scheme~\cite{shoup2000practical,boneh2001short,cachin2005random}.

A large and growing body of literature has investigated how to reduce the word complexity of BA algorithms. 
Recently, Momose and Ren~\cite{momose2020optimal} have presented a synchronous protocol with $O(n^2)$ words, which meets Dolev and Reischuk's long-standing lower bound~\cite{dolev1985bounds}.
Spiegelman~\cite{spiegelman2021search} considered the more common case, where the number of actual failures, denoted by $f$, is smaller than $t$ with resilience of $n=3t+1$. 
In this paper we consider better resilience and ask:

\begin{center}
\textit{
Can we design a BA protocol with $O(n(f+1))$ communication complexity in runs with $f \leq t$ failures, where $n=2t+1$?}
\end{center}

Whereas Dolev and Reischuk's better-known lower bound applies to worst-case runs, they further proved a lower bound of $\Omega(nt)$ signatures in failure-free runs ($f=0$) in a model with a PKI.
At the time, one could have thought that this bound extends to the communication complexity, rendering it $\Omega(nt)$ even with small $f$ values. However, the introduction of threshold signature schemes~\cite{desmedt1987society,shoup2000practical,boneh2001short,cachin2005random} exposed the possibility to compact many signatures into one word, potentially saving many words.

In this paper, we first revisit the original problem as stated in Dolev and Reischuk's work.
In this problem there is a single sender who proposes a value and we refer to this problem as Byzantine Broadcast (BB). We prove that although $O(nt)$ signatures are inevitable, $O(nt)$ messages are not necessary with $f\in o(t)$ failures by presenting an adaptive BB solution with $O(n(f+1))$ words.

The idea behind our algorithm is to reduce this problem to another BA variant. There is a simple reduction from BB to BA with the strong unanimity validity property (from hereon: \emph{strong BA}), which states that if all correct processes propose the same value, this is the only allowed decision. In this reduction, the sender initially sends its value to all other processes who then run a BA solution. Unfortunately, at the time of writing this paper no \emph{adaptive} strong BA was known. I.e., a strong BA solution where communication complexity depends on $f$, rather than on $t$. 
Instead, in Section~\ref{section:abb} we reduce the problem to a new \emph{weak BA} problem with a weaker validity property, \emph{unique validity}, which we define in this paper.
Since the publication of this paper, Elsheimy et al.~\cite{elsheimy2024deterministic} have closed the gap and provided an adaptive strong BA solution, building upon the techniques presented here. We also wish to thank them for identifying an error in our previous version. In this updated version, the weak BA algorithm has been corrected.

Intuitively, the validity condition of weak BA is somewhere between weak unanimity, where if all processes are correct and propose the same value this is the only allowed decision, and external validity~\cite{cachin2001secure}, where a decision value must satisfy some external predicate.
In weak BA, one can define its desired predicate and the requirement is that if all correct processes propose the same value and Byzantine processes cannot devise a value that satisfies the chosen predicate, then the decision must be valid. Otherwise, $\bot$ is allowed.

While the unique validity condition seems to be weak, it is surprisingly powerful when provided the ``right'' external predicate. For example, we can determine that a value is valid if it has at least $t+1$ unique signatures, assuring that some correct process in the system knows this value. 
Unique validity may be of independent interest as a tool for designing algorithms.
We present our adaptive weak BA in Section~\ref{section:aba}. The weak BA, in turn, exploits the quadratic solution by Momose and Ren~\cite{momose2020optimal}.
Figure~\ref{fig:ba_solutions} describes the relation between the various solutions.

\begin{figure}
    \centering
    \includegraphics[scale=0.5]{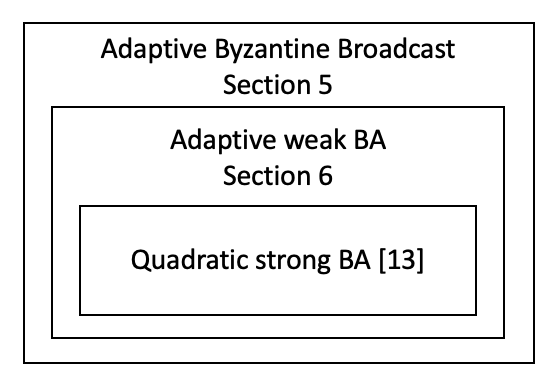}
    \caption{Relation between various Byzantine Agreement solutions. Each box uses the primitives within it.}
    \label{fig:ba_solutions}
\end{figure}

Finally, we consider strong BA. In Section~\ref{section:strong}, we present the first optimally resilient strong binary BA protocol with $O(n)$ communication complexity in the failure-free case. This leaves open the question whether a fully adaptive (to any $f$) strong BA protocol exists.
We summarize the results in Table~\ref{tab:compare}.

\begin{table}[]
\caption{Bounds on communication complexity of deterministic synchronous Byzantine Agreement algorithms with resilience $n=2t+1$.}
\label{tab:compare}

\begin{tabular}{|c|c|c|}
 \hline

  & \textbf{Upper Bound}   & \textbf{Lower Bound}    \\

\hline  
Byzantine Broadcast  &  $O(n(f+1))$ \textbf{Section~\ref{section:abb} + Section~\ref{section:aba}} & $\Omega(nf)$ ($\Omega(n^2)$ signatures)~\cite{dolev1985bounds}  \\
 \hline

Strong BA & $O(n^2)$ multi-valued Momose-Ren~\cite{momose2020optimal}         & $\Omega(nf)$ binary    \\

  & $O(n)$ with $f=0$, binary \textbf{ Section~\ref{section:strong} }    &   ($\Omega(n^2)$ signatures)~\cite{dolev1985bounds} \\

\hline
Weak BA &  $O(n(f+1))$ multi-valued \textbf{Section~\ref{section:aba}} & $\Omega(n)$\\
\hline

\end{tabular}
\end{table}

\section{Model and Preliminaries}

We consider a distributed system consisting of a well-known static set $\Pi$ of $n$ processes and an adaptive adversary. 
The adversary may adaptively corrupt up to  $t<n, n=2t+1$ processes in the course of a run.
A corrupted process is \emph{Byzantine}; it may deviate arbitrarily from the protocol. In particular, it may crash, fail to send or receive messages, and send arbitrary messages. As long as a process is not corrupted by the adversary, it is \emph{correct} and follows the protocol. We denote by $0\leq f\leq t$ the actual number of corrupted processes in a run.

\textbf{Cryptographic tools.}
We assume a trusted public-key infrastructure (PKI) and a computationally bounded adversary. Hence, we can construct and use a threshold signature scheme~\cite{shoup2000practical,boneh2001short,cachin2005random}. We denote by $\lr{m}_p$ the message 
$m$ signed by process $p$.
Using a $(k,n)$-threshold signature scheme, $k$ unique signatures on the same message $m$ can be batched into a threshold signature for $m$ with the same length as an individual signature. 
For simplicity we abstract away the details
of cryptography and assume the threshold signature schemes are ideal.
In practice, our results hold except with arbitrarily small probability, depending on the security parameters.

\textbf{Communication.}
Every pair of processes is connected via a reliable link.
If a correct process $p_i$ receives a message $m$ indicating that $m$ was sent by a correct process $p_j$, then $m$ was indeed generated by $p_j$ and sent to $p_i$. 
The network is synchronous. Namely, there is a known bound $\delta$ on message delays, allowing us to design protocols that proceed in rounds.
Specifically, if a correct process sends
a message to any other correct process at the beginning of some round, it is received by the end of the same round.

\textbf{Complexity.}
We use the following standard complexity notions~\cite{abraham2019asymptotically,spiegelman2021search,momose2020optimal}. While measuring complexity, we say that a \emph{word} contains a constant number of signatures and values from a finite domain, and each message contains at least 1 word.
The communication complexity of a protocol is the maximum number of words sent by all correct processes, across all runs. The \emph{adaptive} complexity is a complexity that depends on $f$.

\section{Problem Definitions}

We consider a family of agreement problems all satisfy agreement and termination defined as follows:

\begin{description}

    \item [Agreement] No two correct processes decide differently.
    \item [Termination] Every correct process eventually decides.
\end{description}

In addition, each variant of the problem satisfies some validity property.
In the Byzantine Broadcast (BB) problem, a designated sender has an input to broadcast to all $n$ processes. The goal is that all correct processes decide upon the sender's value. If the sender is Byzantine, however, it is enough that all correct processes decide upon some common value. Formally,

\begin{definition}[Byzantine Broadcast]
\label{def:BB}
In Byzantine Broadcast, a designated sender $\emph{sender}$ has an input value $v_\emph{sender}$ to broadcast to all processes, and each correct process decides on an output value $decision_i$. BB solution must satisfy agreement, termination and the following validity property:

\begin{description}
    \item [Validity] If \emph{sender} is correct, then all correct processes decide $v_\emph{sender}$.
\end{description}

\end{definition}

Byzantine Agreement (BA) is a closely related problem to BB. In this problem, a set $\Pi$ of $n$ processes each propose an initial value and they all attempt to reach a common decision. In addition, the decided value must be ``valid'' in some sense that makes the problem non-trivial.
 The classic notion of validity states that if all correct processes in $\Pi$ share the same initial value, then the decision must be on this value.
 This property is known as \emph{strong unanimity}, and it entails a limitation on the resilience of a protocol, requiring that $n\geq 2t+1$. 
For hereon we refer to BA with strong unanimity validity condition as \emph{strong BA}.
Formally,

\begin{definition}[Strong Byzantine Agreement]
\label{def:BA}
In Byzantine Agreement, each correct process $p_i\in\Pi$ proposes an input value $v_i$ and decides on an output value $decision_i$.
Any strong BA solution must satisfy agreement, termination and the following validity property:

\begin{description}
        \item[Strong unanimity] If all correct processes propose the same value $v$, then the output is $v$.

\end{description}

\end{definition}

 A different validity property requires that a decision satisfies some external boolean predicate (we call such value a \emph{valid} value). It is used under the assumption that all correct processes propose valid values.
 This is known as \emph{external validity}~\cite{cachin2001secure} and only requires $n>t$. 
 External validity by itself is trivial in case there is a well-known predefined value that satisfies the predicate.
 However, it is commonly used in settings with signatures, where valid values can be verified by all but generated only by specific users or sets thereof.
 For instance, consider a predicate that verifies that $v$ is signed by $n-t$ processes -- no process can unilaterally generate a default valid value.

Our notion of unique validity adopts external validity to allow default values to be decided in cases when there is no unanimous valid value. We say that a value $v$ \emph{exists} in a run of a BA protocol if $v$ is either the input value of a correct process or can be generated by a Byzantine process. E.g., any value signed by a non-Byzantine process cannot be generated locally by a Byzantine process.
Unique validity stipulates that there is a default value if and only if there exists more than one valid value in a BA run. Formally,

\begin{definition}[Weak Byzantine Agreement]
\label{def:wBA}
In weak Byzantine Agreement, each correct process $p_i\in\Pi$ proposes an input value $v_i$ and decides on an output value $decision_i$.
Any weak BA solution must satisfy agreement, termination and the following validity property:

\begin{description}
        \item[Unique Validity] 
        Assume an arbitrary predicate $\textit{validate}(v) \in \{\textit{true}, \textit{false}\}$ that can be computed locally. 
        If a correct process decides $v$ then either $v=\bot$ or $\textit{validate}(v) = \textit{true}$, and if $v=\bot$ then more than one valid value exists in the run.

\end{description}

\end{definition}

As the definition suggests, unique validity is satisfied in weak BA with respect to any chosen external predicate. This allows for the application level to determine the desired properties, and choose the relevant external predicate accordingly. As a simple example, one can think of a predicate that specifies that a value is valid if it is signed by at least $t+1$ processes stating that this value was their initial value. In this scenario, unique validity yields exactly the common strong unanimity property on the underlying signed values.

In fact, unique validity is a useful tool when designing distributed algorithms as it allows to use BA as a framework. Different applications may require different validity conditions, yet still unique validity prevents the system from having a trivial solution in the presence of Byzantine processes.
Note, in addition, that every solution to BA with external validity property immediately solves weak BA.

\section{Related Work}

The starting point of this work goes back to 1985 when Dolev and Reischuk proved two significant lower bounds for the Byzantine Broadcast problem. Specifically, they have studied the worst-case message complexity over all runs and proved it to be $\Omega(nt)$. Moreover, in the authenticated model, which was somewhat undeveloped at the time, they proved a lower bound of $\Omega(nt)$ signatures -- even in a failure-free run.

Since the publication of their fundamental results, the paradigm of complexity measurement has shifted. 
The number of messages is of little importance nowadays, compared to the number of words it entails.
The total number of words (the communication complexity) better reflects the load on the system and is commonly used today when analyzing distributed algorithms.
For example, Dolev and Reischuk presented in their paper a BB algorithm that matches their messages' lower bound.
It requires $O(nt)$ messages, but as a single message can be composed of many different signatures it requires a cubic number of words. It was not until recently that a solution with quadratic communication complexity was presented for synchronous BA with optimal resilience~\cite{momose2020optimal}.

Dolev and Reischuk's complementary lower bound on signatures does not translate to a bound on the communication complexity of an algorithm. Only a few years after Dolev and Reischuk's work, the threshold signature scheme was introduced~\cite{desmedt1987society}. This scheme allows multiple signatures to be compacted into a single combined signature of the same size. That is, a single word can carry multiple signatures.
In this work, we focus on the communication complexity of the BB and BA problems while taking advantage of such schemes. 

To make our algorithms efficient in real-world systems, we adjust the complexity to match the actual number of faults. Moreover, we do so without compromising the worst-case complexity. If all $t$ possibly Byzantine processes crash, the complexity of our algorithms is $O(nt)$. However, in most runs, where systems do not exhibit the worst crash patterns, the complexity is much lower. In fact, it is linear in the number of faults times $n$.

While consensus algorithms were designed to be adaptive in the number of failures over 30 years ago~\cite{dolev1990early}, these works focus on the number of rounds that it takes to reach a decision rather than on communication complexity.
A special case of adaptivity is focusing on failure-free runs. This problem was addressed both by Amdur et al.~\cite{amdur1992message} (only for crash failures) and by Hadzilacos and Halpern~\cite{hadzilacos1993message}. However, both works measure the number of messages rather than words and have sub-optimal communication complexity.

A recent work by Spiegelman~\cite{spiegelman2021search} tackled the problem of adaptive communication complexity in the asynchronous model.
It presents a protocol that achieves correctness in asynchronous runs and requires $O(ft+t)$ communication in synchronous runs. However, due to the need to tolerate asynchrony, its resilience is only $n\ge 3t+1$. This solution relies on threshold signatures schemes, as we do.

As noted also by Momose and Ren~\cite{momose2020optimal}, designing optimally-resilient protocols for the synchronous model limits the use of threshold signatures.
While this primitive has been used in various eventually synchronous and asynchronous works over the last few years~\cite{Abraham2018HotStuffTL,spiegelman2021search, naor2020expected,cachin2005random}, usually with a threshold of $n-t$. Using this threshold in settings with resilience $n=3t+1$, we get certificates signed by at least $t+1$ correct processes.
However, for a resilience of $n=2t+1$, this is no longer the case. The threshold signatures ``lose'' their power as $n-t=t+1$ for which no intersection properties between correct processes signing two distinct certificates can be derived. 
In this work, we exploit threshold signatures with this improved resilience by carefully choosing a better threshold for our needs, as we discuss in Section~\ref{section:aba}. We mention that although not using threshold schemes, Xiang et al.~\cite{xiang2021strengthened} also benefit from collecting more than $n-t$ signatures in some scenarios.

\section{From Weak BA to Adaptive Byzantine Broadcast}
\label{section:abb}

In this section we study the BB problem, and optimize its adaptive communication complexity over all runs.
We present a new BB protocol with resilience $n=2t+1$ and adaptive communication complexity of $O(n(f+1))$.

Recall that in the BB problem there is only one sender who aims to broadcast its initial value and have all correct processes agree on it. If the sender is Byzantine, it may attempt to cause disagreement across correct processes.
There is a known simple and efficient reduction from BB to strong BA. Given a strong BA solution, the designated sender starts by sending its value to all processes, and then they all execute the BA solution and decide on its output. It is easy to see that if the sender is correct, all correct processes begin the strong BA algorithm with the same input, and by strong unanimity they then decide upon the sender's value. 

However, trying to apply the same reduction from BB to weak BA no longer works.
If the sender is Byzantine, the correct processes do not have a valid initial value for the BA.
Nonetheless, in this section we present a reduction from BB to weak BA\footnote{This reduction only works if $n\geq2t+1$.}, which incurs a cost of $O(n(f+1))$ words. Thus, together with an adaptive weak BA with the same complexity, we obtain a synchronous adaptive BB algorithm with a total of $O(n(f+1))$ words and resilience $n=2t+1$. At this point we 
assume that such adaptive weak BA is given as a black box. An implementation for this primitive is presented in Section~\ref{section:aba}.

\begin{algorithm}
    \begin{algorithmic}[1]
    \Statex Initially $v_i,\emph{val},\emph{decision},\emph{ba\_decision}=\bot$ 
    \Statex

    \Statex \textbf{Round 1:}
    \If{$\emph{sender}=p_i$} \label{l:first_line}
        \State send $\lr{v_{sender}}_{sender}$ to all \label{l:initial_s_v}
    \EndIf
    
    \If{received message $\lr{v}_{sender}$ from \emph{sender}}
        \State $v_i\gets \lr{v}_{sender}$ \label{l:rec_init_v}

    \EndIf
    
    \For {$j=1$ to $n$} \label{l:start_for}
            \State $val\gets invokePhase(j,v_i)$
            \If{$val\neq \bot$}
                \State $v_i\gets val$ \label{l:rec_val_v}
            \EndIf

    \EndFor \label{l:end_for}
    
    \State $\emph{ba\_decision} \gets$ weak BA with \emph{BB\_valid} predicate and initial value $v_i$\label{l:decision}
    
    \If{$\emph{ba\_decision}$ is of the form $\lr{v}_{sender}$}
    \State $\emph{decision} \gets v$ \label{l:extract_v_s}
    \Else
    \State $\emph{decision} \gets \bot$ \label{l:eoa}
    \EndIf

     \algstore{bb2}

    \end{algorithmic}
    \caption{BB algorithm: code for process $p_i$, \emph{sender}'s input is $v_{\emph{senedr}}$}
    \label{alg:BB_main}
\end{algorithm}

Our algorithm, presented in Algorithms~\ref{alg:BB_main} and~\ref{alg:BB_phase}, is composed of three parts. 
The first part (lines~\ref{l:first_line}~--~\ref{l:rec_init_v} in Algorithm~\ref{alg:BB_main}) is the first round in which the leader disseminates its value. Processes that receive that value adopt it as their BA initial value (line~\ref{l:rec_init_v}).
The second part (lines~\ref{l:start_for}~--~\ref{l:end_for} in Algorithm~\ref{alg:BB_main} and Algorithm~\ref{alg:BB_phase}) is a ``vetting'' part. It consists of $n$ phases, with a rotating leader.
Leaders initiate phases to learn about the first part's initial value.
Finally, the third part (lines lines~\ref{l:decision}~--~\ref{l:eoa} in Algorithm~\ref{alg:BB_main}) is a weak BA execution.

Deciding upon the weak BA output takes care of the agreement and termination properties. It is left to (1) satisfy the BB validity property and (2) make sure that the preconditions for the weak BA hold, that is, each correct process has a valid input to propose. To achieve these properties, we define the
$\emph{BB\_valid}(v)$ predicate in the following way. 
$\emph{BB\_valid}(v)=\emph{true}$ if and only if $v$ is signed by either the sender or by $t+1$ processes.

Note that if the sender happens to be Byzantine, it is acceptable to decide on any value. However, it is important to make sure that if the sender is correct, then the only valid value is its initial BB input.
Simply setting a value to be valid only if it is signed by the sender would not work, as it allows a faulty sender to cause a scenario in which there are no valid values to agree upon by not sending its value to any process.
Note that we cannot simply fix this by introducing some default valid value: If we were to do so, it would be valid to agree on that value also in the case of a correct sender, violating the BB validity condition.

Our algorithm makes sure that if the sender is correct, the second condition in the $\emph{BB\_valid}$ definition cannot be satisfied, and hence there is only one possible outcome to the BA algorithm. However, if the sender is Byzantine, it is guaranteed that there is some value to decide upon. That is, all correct processes start the weak BA with an initial value that satisfies the predicate.

In the vetting part of the algorithm, we ensure that the above-mentioned conditions hold. Moreover, we do so with a communication complexity that is adaptive to the number of actual process failures.
The core idea is to work in leader-based \emph{phases}. Every phase has a unique leader and is composed of a constant number of leader-to-all and all-to-leader synchronous rounds. Every phase is initiated by a leader-to-all message.
If the leader decides not to send the initial message then no messages will be sent by correct processes in this phase and we say that this phase is \emph{silent}, and otherwise, it is \emph{non-silent}. In our algorithm, a phase is non-silent if the phase's leader did not choose an initial value for the BA prior to that phase.


In every phase, each process $p_i$ starts the phase with some initial value $v_i$ and if the phase is non-silent it returns some value. The requirements from the phase are: (1) If the phase's leader is correct and the phase is non-silent, then all correct processes return a valid value. (2) All correct processes return either $\bot$ or a valid $v$. And (3) if the sender is correct, then no correct process returns a value signed by $t+1$ processes.

Upon a non-silent phase, the leader starts by asking all processes for help by sending a $\helpreq$ message (line~\ref{l:help_bb_2}). A correct process that receives a help request message answers the leader.
If it has set a BA initial value, it sends it to the leader at line~\ref{l:reply_v}, and otherwise, it sends a signed $\idk$ (i don\textquotesingle t know) message at line~\ref{l:reply_idk}.
If the leader receives a value signed by the designated sender it broadcasts it (line~\ref{l:bc_v}).
Otherwise, if it receives $t+1$ $\idk$ messages, it uses a threshold signature scheme to create an $\idk$ quorum certificate and broadcasts it (line~\ref{l:bc_idk}).
A process that receives from the leader a value signed by either the sender or any $t+1$ processes returns it. Otherwise, it returns $\bot$.

At the end of each non-silent phase, a correct process that returns a $v\neq\bot$ from the phase, updates its local $v_i$ accordingly at line~\ref{l:rec_val_v}. This value at the end of the $n^{th}$ phase is the input for the weak BA algorithm.
Since we execute $n$ phases, all correct processes set valid values by the end of all phases.
This is because once there is a correct process that did not set a value it initiates its phase and then all correct processes return with a valid value.
At this point, all processes execute the weak BA and decide upon its output (line~\ref{l:decision}).

\begin{algorithm}
    \begin{algorithmic}[1]
  \algrestore{bb2}
    \State \emph{leader}$\gets p_{j\mod n}$
    \Statex \textbf{Round 1:}
    \If{$\emph{leader}=p_i$ and $v_i=\bot$}
        \State broadcast the message $\lr{\helpreq,j}_\emph{leader}$ \label{l:help_bb_2}
    \EndIf

    \Statex \textbf{Round 2:}
    \If{received $\lr{\helpreq,j}_\emph{leader}$}
    
        \If{$v_i\ne \bot$}
        \State send $\lr{v_i,j}$ to \emph{leader} \label{l:reply_v}
    \Else 
        \State send $\lr{\idk,j}_{p_i}$ to \emph{leader} \label{l:reply_idk}
    \EndIf
    \EndIf

    \Statex \textbf{Round 3:}
    \If{$\emph{leader}=p_i$}
    
    \If{received $\lr{v',j}$ s.t. $v'=\lr{v}_{sender}$} \label{l:r_s_v}

        \State broadcast the message $\lr{\lr{v}_{sender},j}$ \label{l:bc_v}

    \ElsIf{received $t+1$ unique signatures $\lr{\idk,j}_{p'}$}
        \State batch these messages into $QC_{\idk}$ using a $(t+1,n)$-threshold signature scheme
        \State broadcast the message $\lr{QC_{\idk},j}$ \label{l:bc_idk}
    \EndIf

    \EndIf

    \If{received $\lr{v,j}$ from \emph{leader} and \emph{BB\_valid}$(v)=$\emph{true}}
        \State return $v$ \label{l:return_valid}
    \Else
        \State return $\bot$ \label{l:return_bottom}
    \EndIf

    \end{algorithmic}
    \caption{$invokePhase(j,v_i)$: code for process $p_i$}
    \label{alg:BB_phase}
\end{algorithm}

A formal correctness proof of Algorithms~\ref{alg:BB_main} and~\ref{alg:BB_phase} appears in Appendix~\ref{bb_correct}, proving the following theorem:

\begin{theorem}
Algorithm~\ref{alg:BB_main} solves BB.
\end{theorem}



\subsection{Complexity}

We prove that the complexity of Algorithms~\ref{alg:BB_main} and~\ref{alg:BB_phase} is $O(n(f+1))$.

Each non-silent phase is composed of a constant number of all--to--leader and leader--to--all rounds and thanks to the use of threshold signatures, all messages sent have a size of one word. Thus, each phase incurs $O(n)$ words. In total, there are potentially $n$ phases.
However, we prove in Appendix~\ref{bb_correct}
that after the first non-silent phase by a correct leader, all following phases with correct leaders are silent.
Thus, the number of non-silent phases is linear in $f$.
We conclude that all phases in lines~\ref{l:start_for}~--~\ref{l:end_for} use $O(n(f+1))$ words. The complexity of the weak BA black box is also
$O(n(f+1))$ (as we will show in the next section), resulting in a total of $O(n(f+1))$ words.

\section{Adaptive Weak BA}
\label{section:aba}

In this section, we present a synchronous adaptive weak BA algorithm with resilience $n=2t+1$.
This algorithm is the missing link for the adaptive BB presented in the previous section. 
Once again, we use the concept of phases and exploit the pattern of possible \emph{silent} phases. In this algorithm, the phases are slightly different and the decision to start a phase as a leader depends on whether or not the leader has reached a decision in previous phases.

Unlike the BB problem, in BA every process begins the algorithm with its own input value.
Communication-efficient solutions to this problem usually employ threshold signatures schemes~\cite{Abraham2018HotStuffTL,spiegelman2021search}.
This technique is widely used in asynchronous and eventually synchronous protocols, with resilience $n=3t+1$. In these contexts, one can use a scheme of $(n-t)$-out-of-$n$ signatures and benefit from the fact that any two such quorum certificates intersect by at least $t+1$ processes, and therefore at least one correct process.

Unfortunately, when trying to apply the same technique to a system with resilience $n=2t+1$, it fails. A correct process might be unable to obtain $2t+1$ unique signatures on any value as Byzantine processes might not sign it. On the other hand, a quorum certificate with only $t+1$ unique signatures is not very useful as it does not guarantee the desired intersection property.

Our first key observation is that the intersection property can be achieved as long as we have $\ceil[\big]{\frac{n+t+1}{2}}$ unique signatures.
If we obtain this number of signatures out of $n=2t+1$, safety is preserved in the sense that conflicting certificates cannot be formed by a malicious adversary. Of course, there are runs in which we cannot reach that threshold since $\ceil[\big]{\frac{n+t+1}{2}}>n-t$ (e.g., if $t$ processes crash immediately as the run begins).
But in this case, $f\ge \frac{t}{2}$, and $O(f)$ becomes asymptotically $O(t)$. Hence, we can use a fallback algorithm with $O(nt)$ communication complexity. 

As we assume that $t\in \Theta(n)$, we can use Momose and Ren's synchronous algorithm that has $O(n^2)$ communication complexity~\cite{momose2020optimal} for the fallback. We denote that algorithm $\mathcal{A}_\emph{fallback}$. Note that their algorithm is ``stronger'' than our proposed algorithm as it provides strong unanimity for validity (i.e., it solves strong BA).
We can use their solution by checking the validity of $\mathcal{A}_\emph{fallback}$'s output according to the predicate. If it is valid, this is the decision value, and otherwise a default valid value is decided.
Equipped with these insights, we next present our algorithm.

During the phases part of the protocol, a correct process must commit a value before reaching a decision. When it has certainty about a value it updates that value in a \emph{commit} variable, along a \emph{commit\_proof} of this commitment (a quorum certificate, signed by sufficiently many processes) and a \emph{commit\_level} indicating the latest phase of a valid commitment it heard of. Once a correct process commits to a certain value it can only commit to a value for which it heard a valid commitment proof in a later phase during the run. That is, it may decide on a value for which it did not send a commit message. Moreover, it may even decide on a value it did not commit at all. For example, if it reaches the fallback and no correct process has decided. Once a correct process reaches a decision it updates it in its local \emph{decision} variable as well as a matching quorum certificate in \emph{decide\_proof} variable.

\algdef{SE}[SUBALG]{Indent}{EndIndent}{}{\algorithmicend\ }%
\algtext*{Indent}
\algtext*{EndIndent}

\newcommand{\finalized}{\mathsf{finalized}\xspace}

\algdef{SE}[UPON]{Upon}{EndUpon}[1]{\textbf{upon
}\ #1\ \algorithmicdo}{\algorithmicend\ \textbf{}}%
\algtext*{EndUpon}

\begin{algorithm}
    \begin{algorithmic}[1]
    
    \Statex Initially $\emph{decision}=\emph{undecided}, \emph{bu\_decision}=v_i, \emph{fallback\_start}\gets \infty$
    \Statex $\emph{decide\_proof}, \emph{commit}, \emph{commit\_proof},\emph{bu\_proof},\emph{fallback\_val},\emph{phase\_decision}=\bot$ 
    \Statex $\emph{commit\_level}\gets 0$
    \Statex

    \For {$j=1$ to $t+1$}

            \State $\emph{phase\_decision}, \emph{decide\_proof}, \emph{commit}, \emph{commit\_proof},\emph{commit\_level}\gets
            $
            \Statex 
            \hskip1.5em $\emph{invokePhase}(j,v_i,decision, commit, commit\_proof,commit\_level)$ 
            
            \If{$\emph{decision}=\emph{undecided} $ and $\emph{phase\_decision}\neq \emph{undecided}$}
                \State $\emph{decision} \gets \emph{phase\_decision}$ \label{l:inv_phase}
            \EndIf

    \EndFor \label{line:end_for}

    \Statex \textbf{Round 1:}

    \If{$decision=\emph{undecided}$} \label{l:end_phases}
        \State broadcast $\lr{\helpre}_{p_i}$\label{l:help}
    \EndIf
    
    \Statex \textbf{Round 2:}
    \If{received $\lr{\helpre}_{p'}$ message and $\emph{decision}\neq \emph{undecided}$}\label{l:r_help}
        \State send $\lr{\help,\emph{decision},\emph{decide\_proof}}_{p_i}$ to $p'$\label{l:a_help}
    \EndIf
    
    \If{received $t+1$ messages of $\lr{\helpre}_{p'}$ from different processes}
    
        \State batch these messages into $QC_{\fb}(v)$ using a $(t+1,n)$-threshold signature scheme
    
        \State broadcast the message $\lr{\fb,QC_{\fb},\emph{decision},proof}_{p_i}$  \label{l:bc_dec1}
        \State $\emph{fallback\_start}\gets \emph{now}+2\delta$

    \EndIf
    
    \Statex \textbf{Round 3:}

    \If {received $\lr{\help,v,\emph{decide\_proof}}_{p'}$ with valid $v$ and \emph{decide\_proof} for $v$ and $\emph{decision}=\emph{undecided}$}
        \State $\emph{decision}\gets v$\label{l:got_help}
    \EndIf

     \State $\emph{bu\_decision}\gets decision$ 

    \While{$\emph{fallback\_start}>\emph{now}$}
        \If {received valid $\lr{\fb,QC_{\fb},v,proof_{p'}}_{p'}$} \label{l:safety_wind}
            \If{$\emph{decision}=\emph{undecided}$ and $proof_{p'}\neq \bot$ is a valid proof for a valid $v$}
        \State $\emph{bu\_decision}\gets v$ \label{l:got_fb}
        \State $bu\_proof\gets proof_{p'}$
        \EndIf

        \If{$\emph{fallback\_start}= \infty$}
            \State broadcast the message $\lr{\fb,QC_{\fb},\emph{bu\_decision},bu\_proof}_{p_i}$ \label{l:bc_fb}
            \State $\emph{fallback\_start}\gets \emph{now}+2\delta$

        \EndIf

        \EndIf
        
    \EndWhile
    
    \State $\emph{fallback\_val}\gets \mathcal{A}_{\emph{fallback}}$ with $\delta'=2\delta$ and initial value \emph{bu\_decision} \label{l:fb} 
    
    \If{$\emph{decision}=\emph{undecided}$}
        \If{\emph{fallback\_val} is valid} \label{l:pos_r1}
            \State $\emph{decision}\gets \emph{fallback\_val}$ \label{l:final_dec}
        \Else
            \State $\emph{decision}\gets \bot$\label{l:ret_def}
        \EndIf \label{l:pos_r2}
        
    \EndIf

     \algstore{part1}

    \end{algorithmic}
    \caption{weak BA algorithm: code for process $p_i$ with initial value $v_i$}
    \label{alg:BA_main}
\end{algorithm}

\begin{algorithm}
    \begin{algorithmic}[1]
   \algrestore{part1}
    \State \emph{leader}$\gets p_{j\mod n}$

    \Statex \textbf{Round 1:}
    \If{$\emph{leader}=p_i$ and $\emph{decision}=\bot$}
        \State broadcast the message $\lr{\propose, v_i,j}_{\emph{leader}}$ \label{l:first_bc}
    \EndIf

    \Statex \textbf{Round 2:}
    \If{received $\lr{\propose, v,j}_{\emph{leader}}$ with a valid $v$ for the first time and $commit=\bot$}\label{l:valid_v}
        \State send $\lr{\vote, v,j}_{p_i}$ to $\emph{leader}$ \label{l:sign_vote}
    \ElsIf{received $\lr{\propose, v,j}_{\emph{leader}}$ and $commit\neq \bot$}
        \State send $\lr{\commit,\emph{commit,commit\_proof},\emph{commit\_level},j}_{p_i}$ to $\emph{leader}$ \label{l:first_commit}
    \EndIf

    \Statex \textbf{Round 3:}
    \If{$\emph{leader}=p_i$}
        \If{received $\lr{\commit, w, QC_{\commit}(w),level_{\commit}, j}_{p'}$}
            \State broadcast the message $\lr{\commit,w, QC_{\commit}(w),level_{\commit}, j}_{\emph{leader}}$  according to the maximal $level_{\commit}$ received

        \ElsIf{received $\ceil[\big]{\frac{n+t+1}{2}}$ messages of $\lr{\vote, v, j}_{p'}$} \label{l:thresh_votes}
            \State batch these messages into $QC_{\commit}(v)$ using a $(\ceil[\big]{\frac{n+t+1}{2}},n)$-threshold signature scheme
            \State broadcast the message $\lr{\commit,v, QC_{\commit}(v), j, j}_{\emph{leader}}$
        \EndIf
    \EndIf
    
    \Statex \textbf{Round 4:}
    \If{received $\lr{\commit,v, QC_{\commit}(v), level_{\commit}, j}_{\emph{leader}}$ and $level_{\commit}\geq \emph{commit\_level}$ and $level_{\commit}$ is valid according to $QC_{\commit}(v)$}\label{line:checklevel}
    
    \State send $\lr{\decide, v, j}_{p_i}$ to \emph{leader}\label{l:sign_dec}

        \State $\emph{commit}\gets v$
        \State \emph{commit\_proof}$\gets QC_{\commit}(v)$
        \State \emph{commit\_level}$\gets level_{\commit}$ \label{line:rank}
        \EndIf


    \Statex \textbf{Round 5:}
    \If{$\emph{leader}=p_i$}
        \If{received $\ceil[\big]{\frac{n+t+1}{2}}$ messages of $\lr{\decide, v, j}_{p'}$}
            \State batch these messages into $QC_{\finalized}(v)$ using a $(\ceil[\big]{\frac{n+t+1}{2}},n)$-threshold signature scheme

            \State broadcast the message $\lr{\finalized,v, QC_{\finalized}(v), j}_{\emph{leader}}$

    \EndIf
    \EndIf

    \If{received $\lr{\finalized,v, QC_{\finalized}(v), j}_{\emph{leader}}$}

        \State $\emph{decision}\gets v$ \label{l:set_decision}
        \State \emph{decide\_proof}$\gets QC_{\finalized}(v)$

    \EndIf

    \State return $(\emph{decision}, \emph{decide\_proof}, \emph{commit}, \emph{commit\_proof},\emph{commit\_level})$\label{l:phase_return}

    \end{algorithmic}
    \caption{$invokePhase(j, v_i, decision, \emph{decide\_proof}, commit, commit\_proof,commit\_level)$: code for process $p_i$}
    \label{alg:invoke_phase}
\end{algorithm}

\textbf{A single phase}
The code for a single phase is given in Algorithm~\ref{alg:invoke_phase}. 
Each process $p_i$ starts a phase with its initial value $v_i$ and information about possible previous commits ($\emph{commit, commit\_proof, commit\_level}$) and decisions ($\emph{decision, decide\_proof}$). Correct processes return with updated information about commits and decisions that were made in that phase (or prior to that). The guarantees of the phases are:
(1) Every \emph{decision} updated during a phase is valid; (2) All decisions updated by correct processes are the same and there exists at most one valid \emph{decide\_proof} in the system; and (3) If the phase's leader is correct, the phase is non-silent, and $n-f>\ceil[\big]{\frac{n+t+1}{2}}$, then all correct processes return with the same valid decision.

Every non-silent phase starts with the leader broadcasting a $\propose$ message with its value in line~\ref{l:first_bc}. Upon receiving this message, correct processes either vote for this value by signing it (line~\ref{l:sign_vote}) or answer with a value that was previously committed as well as its commit quorum certificate and level (line~\ref{l:first_commit}). If the leader receives a committed value it simply broadcasts it. Otherwise, if it manages to achieve the required $\ceil[\big]{\frac{n+t+1}{2}}$ threshold of signatures, it can form a quorum certificate committing its proposed value (line~\ref{l:thresh_votes}).

Note that at this point the committed value is not ``safe enough'' to be decided by correct processes. Byzantine leaders may cause correct processes to participate in forming a commit certificate for more than one value. As correct processes that have decided do not initiate phases, they might never communicate without going through Byzantine leaders. Thus, we need another level of certainty, in the form of the finalize certificate (to be stored in \emph{decide\_proof}). Using the commit levels, we maintain the invariant that if a correct process receives a valid $\finalize$ certificate, then no $\finalize$ certificate on another value can be formed.

Thus, after a correct process learns about a committed certificate with a level higher than the previous commit, it sends a matching $\decide$ message to the leader (line~\ref{l:sign_dec}) and updates the commit information accordingly.
If the leader receives the necessary threshold of $\decide$ messages, it forms a $\finalize$ quorum certificate. Every process that receives such a certificate can safely return the certificate's value as its decision.

\textbf{Main algorithm}
The BA algorithm is given in Algorithm~\ref{alg:BA_main}, using the phase algorithm as a building block.
In our algorithm, all correct processes eventually decide by updating their \emph{decision} variable. However, they do not halt.
In our BA algorithm, we start by executing $n$ phases with a rotating leader, ensuring that every correct process has a chance to reach a decision before executing the fallback algorithm.
After the phases end there are several possibilities. First, if there are at most $\frac{n-t-1}{2}$ Byzantine processes, all correct processes must have decided. If there are more Byzantine processes, it may be the case that some correct processes decided and others did not.
This could happen, for example, if a Byzantine leader causes the single correct leader to decide and not initiate its phase. By the phase guarantees, we know that all correct processes that decide by this point, decide the same valid value.

To address the case where not all correct processes decided, we have processes that have not decided ask for help from all other processes (line~\ref{l:help}). If a correct process has decided and receives a $\helpre$ message, it answers with a help message including the decision value along with its proof at line~\ref{l:a_help}. Note that in this round, the number of messages sent by correct processes is linear in the number of help requests. Specifically, if only Byzantine processes send $\helpre$ messages, the number of answers is $O(nf)$ and independent of $t$.

We note that if $t+1$ help requests are sent, then at least one of them is sent by a correct process that did not manage to form quorum certificates when it served as leader.
Thus, in this case, $f\in \Theta(t)$, and we can execute the fallback algorithm. To make sure that all correct processes participate in the fallback algorithm, a $\fb$ certificate with $t+1$ signature is formed.

We now encounter a new challenge.
We must have all correct processes start a synchronous fallback algorithm at the same time. However, an adversary can form the $\fb$ certificate and deal it to only some correct processes. 
This scenario can happen, for example, if less than $t+1$ $\helpre$ messages are sent, and the adversary adds $t$ $\helpre$ signatures of its own.
We thus require a correct process that receives a $\fb$ certificate to broadcast it (line~\ref{l:bc_fb}).
This ensures that whenever one correct process runs the fallback algorithm, all of them do, but may still cause different correct processes to start the fallback at different times.
Nevertheless, we know that the starting time difference is at most the $\delta$ it takes the message to arrive.
We therefore run the fallback algorithm with $\delta'=2\delta$, ensuring that all correct processes enter a fallback round before any of them exits from it.

Another subtle point is making sure that the fallback algorithm does not output a decision value that contradicts previous decisions made by correct processes.
For that reason we add another $2\delta$ safety window between getting notified about a fallback and initiating it. Correct processes that broadcast the $\fb$ certificate attach their decision value and a proof (if exists). In the $2\delta$ safety window, processes that learn about a decision value in the system adopt it as the initial value for the fallback algorithm (line~\ref{l:safety_wind}).
Recall that $\mathcal{A}_\emph{fallback}$ is a strong BA protocol. If a correct process decides $v$ prior to the fallback algorithm, all other correct processes learn about 
$v$ during the safety window. Then, by strong unanimity, they all decide $v$.

Note that if the decision returned from $\mathcal{A}_\emph{fallback}$ is not valid then it must be that strong unanimity preconditions are not satisfied (since correct processes always have valid inputs) and a default value is returned.
Furthermore, whenever the strong unanimity precondition is not satisfied, it follows that not all correct processes propose the same value. As a result, there must exist more than one valid value in the run (the different correct proposals). 
And the $\bot$ default value is a valid weak BA output.


%



A formal correctness proof of Algorithms~\ref{alg:BA_main} and~\ref{alg:invoke_phase} appears in Appendix~\ref{weak_ba_correct}, proving the following theorem:

\begin{theorem}
Algorithm~\ref{alg:BA_main} solves weak BA.
\end{theorem}

\subsection{Complexity}

We show that if $f<\frac{n-t-1}{2}$, correct processes never perform the fallback algorithm. 

\begin{lemma}
If $f<\frac{n-t-1}{2}$, correct processes never perform the fallback algorithm.
\end{lemma}
\begin{proof}

In Appendix~\ref{weak_ba_correct} we prove that if a correct process is the leader of a non-silent phase and $f<\frac{n-t-1}{2}$, then all correct processes return the same valid decision.
Since Algorithm~\ref{alg:BA_main} is composed of $n$ phases, every correct process has a chance to invoke its phase and all correct processes decide by line~\ref{line:end_for}.
Assume by way of contradiction that there exists a correct process that invokes the fallback algorithm. By the code, it has received a $\fb$ certificate. However, such certificate can only be formed by $t+1$ unique $\helpre$ signatures, meaning that at least one correct process sent a $\helpre$ message. But this is impossible if all correct processes decide by line~\ref{line:end_for}.

\end{proof}

Each phase is composed of a constant number of all--to--leader and leader--to--all rounds. Thus, it incurs $O(n)$ words. Potentially, there are $n$ phases. However, a lemma in Appendix~\ref{weak_ba_correct} proves that once a correct leader invokes $invokePhase()$ and the number of actual failures is $f<\frac{n-t-1}{2}$,  all correct processes decide by the end of that phase. Since correct leaders that had already decided do not invoke their phases (their phases are silent), the number of invoked phases depends on $f$ itself.
Thus, all phases combined send $O(n(f+1))$ words.

After $n$ invokePhase invocations end, help request messages are sent only by correct processes that did not decide. By the above-mentioned lemma, it happens only if $f>\frac{n-t-1}{2}$.
In this case, $f=\Theta(n)$ and since $t=\Theta(n)$ it holds that $O(nf)=O(n^2)$.
Correct processes that decide by this point answer directly to whoever sent them help requests, without affecting the asymptotic complexity.
If some correct process receives a $\fb$ certificate, another all-to-all round is added, keeping the complexity $O(n^2)$. All other communication costs are incurred in the fallback algorithm, whose complexity is also $O(n^2)$.

\section{strong BA: the failure free case}
\label{section:strong}

Recall that the optimal resilience for strong BA is $n=2t+1$.
In this section we present a binary strong BA protocol that has communication complexity of $O(n)$ in the failure free case. Otherwise, it has complexity $O(n^2)$.
The question of whether an adaptive protocol with $O(n(f+1))$ complexity can be designed for strong BA with optimal resilience remains open.

\algdef{SE}[SUBALG]{Indent}{EndIndent}{}{\algorithmicend\ }%
\algtext*{Indent}
\algtext*{EndIndent}

\algdef{SE}[UPON]{Upon}{EndUpon}[1]{\textbf{upon
}\ #1\ \algorithmicdo}{\algorithmicend\ \textbf{}}%
\algtext*{EndUpon}

\begin{algorithm}

    \begin{algorithmic}[1]
    
        \Statex Initially $\emph{decision}, \emph{proof},\emph{bu\_decision},\emph{bu\_proof},\emph{fallback\_val}=\bot$ 
        \Statex $\emph{fallback\_start}\gets \infty$

    \Statex

    \State $\emph{leader}\gets p_1$  

    \Statex \textbf{Round 1:}
    \State send $\lr{v_i}_{p_i}$ to \emph{leader} \label{ll:send_v}
    
    \Statex \textbf{Round 2:}
    \If{$\emph{leader}=p_i$}
        \If{received $t+1$ messages of $\lr{v}_{p'}$ for some $v$} 
            \State batch these messages into $QC_{\propose}(v)$ using a $(t+1,n)$-threshold signature scheme
            \State broadcast the message $\lr{\propose,v, QC_{\propose}(v)}_{\emph{leader}}$ \label{ll:bc_p_c}
        \EndIf
    \EndIf
    
    \Statex \textbf{Round 3:}
    \If {received valid $\lr{\propose,v, QC_{\propose}(v)}_{\emph{leader}}$}
        \State send $\lr{\decide,v}_{p_i}$ to \emph{leader} \label{ll:s_d}
    \EndIf
    
    \Statex \textbf{Round 4:}
    \If{$\emph{leader}=p_i$}
        \If{received $n$ messages of $\lr{\decide,v}_{p'}$} 
            \State batch these messages into $QC_{\decide}(v)$ using a $(n,n)$-threshold signature scheme
            \State broadcast the message $\lr{\decide,v, QC_{\decide}(v)}_{\emph{leader}}$ \label{ll:bc_d_c}
        \EndIf
    \EndIf
    
    \Statex \textbf{Round 5:}
    \If {received valid $\lr{\decide,v, QC_{\decide}(v)}_{\emph{leader}}$ and $\emph{decision}=\bot$}
        \State \emph{decision}$\gets v$ \label{ll:dec}
        \State \emph{proof}$\gets QC_{\decide}(v)$
    \Else
        \State broadcast the message $\lr{\fb, \bot, \bot}_{p_i}$  \label{ll:fall}

            \State $\emph{fallback\_start}\gets \emph{now}+2\delta$

    \EndIf
 
     \State $\emph{bu\_decision}\gets decision$

    \While{$\emph{fallback\_start}>\emph{now}$}
        \If {received $\lr{\fb,v,proof_{p'}}_{p'}$}
            \If{$\emph{decision}=\bot$ and $proof_{p'}\neq \bot$ is a valid proof for a valid $v$}
        \State $\emph{bu\_decision}\gets v$ \label{ll:bu}
        \State $bu\_proof\gets proof_{p'}$
        \EndIf
        
        \If{$\emph{fallback\_start}= \infty$}
            \State broadcast the message $\lr{\fb,\emph{bu\_decision},bu\_proof}_{p_i}$ 
            \State $\emph{fallback\_start}\gets     \emph{now}+2\delta$

        \EndIf

        \EndIf
        
    \EndWhile

    \State  $\emph{fallback\_val}\gets \mathcal{A}_{\emph{fallback}}$ with $\delta'=2\delta$ and initial value \emph{bu\_decision}  \label{ll:fb}
    
    \If{$\emph{decision}=\bot$}
        \State $\emph{decision}\gets\emph{fallback\_val}$ \label{ll:dec_fb_v}
    \EndIf

    \end{algorithmic}

    \caption{strong BA algorithm: code for process $p_i$ with initial value $v_i$}
    \label{alg:FF_strong_BA}
\end{algorithm}

In the algorithm, presented in Algorithm~\ref{alg:FF_strong_BA}, a single leader first collects all initial values. Since we solve binary agreement, in the failure-free case there must be a value proposed by $t+1$ different processes. Thus, the leader can use a threshold signature scheme to aggregate a quorum certificate on this proposed value.

As a second step, the leader sends this certificate to all processes and attempts to collect $n$ different signatures on the value. If it succeeds, it broadcasts it. Every process that receives a signed-by-all certificate can safely decide upon its value.
If a correct process does not decide, it broadcasts a fallback message. Every process that hears such a message, echoes it at most once, and execute $\mathcal{A}_{\emph{fallback}}$ after $2\delta$ time with $2\delta$-long rounds, as in Section~\ref{section:aba}.
In Appendix~\ref{strong_ba_correct} we prove the correctness of Algorithm~\ref{alg:FF_strong_BA}. That is, we prove the following theorem:

\begin{theorem}
Algorithm~\ref{alg:FF_strong_BA} solves binary strong BA.
\end{theorem}

\subsection{Complexity}
We show that if the run is failure-free, correct processes never perform the fallback algorithm. 

\begin{lemma}
\label{lem:all_correct}
If $f=0$, correct processes never perform the fallback algorithm.
\end{lemma}
\begin{proof}
If all processes are correct then they all send their initial values to the leader at line~\ref{ll:send_v}. Since values are binary, and there are $n=2t+1$ processes, there must be a value $v$ such that the leader receives $t+1$ unique signatures on $v$. Then, the leader broadcasts a $\propose$ certificate on $v$ (line~\ref{ll:bc_p_c}). Every correct process that receives this certificate replies with a signed $\decide$ message at line~\ref{ll:s_d}. Since all processes are correct, the leader then receives $n$ signatures and then broadcasts a $\decide$ certificate on $v$ (line~\ref{ll:bc_d_c}). All processes then receive this certificate and decide $v$ at line~\ref{ll:dec}. None of them sends a fallback message.
\end{proof}

 By Lemma~\ref{lem:all_correct}, if all processes are correct then they never perform the fallback algorithm, and there are 4 all-to-leader and leader-to-all rounds, with a total of $O(n)$ words. Otherwise, the complexity is the complexity of the fallback algorithm, which is $O(n^2)$.

\section{Conclusions and Future Directions}

We have presented solutions for both Byzantine Broadcast and weak Byzantine Agreement with adaptive communication complexity of $O(n(f+1))$ and resilience $n=2t+1$. 
To construct the weak BA algorithm, we utilized a threshold on the number of signatures such that on one hand, this number is sufficient to ensure a safe algorithm with adaptive communication in case there are not ``many'' Byzantine processes. On the other hand, failing to achieve this threshold indicates that there is a high number of failures, which allows the use of a quadratic fallback algorithm.


This weak BA algorithm is taken as a black box to construct our adaptive BB algorithm.
Here, we carefully choose the predicate for the validity property, to allow us to reduce one problem to the other. Finally, for strong BA we propose a binary solution with optimal resilience. Our solution is linear in $n$ in the practically common failure-free case, and quadratic in any other case. The question of whether a fully adaptive strong BA with optimal resilience exists or not remains open.

While $n=2t+1$ is optimal for strong BA, this is not the case for BB and weak BA, where any $t<n$ can be tolerated\footnote{For weak BA, this stems from the resilience for external validity.}. Thus, another possible future direction is improving the resilience of an adaptive BB or adaptive weak BA to support any $t<n$.
Our weak BA algorithm relies on the current resilience to satisfy that if $f>n-\ceil[\big]{\frac{n+t+1}{2}}$ then $f$ is linear in $t$.
Note that this remains true for any resilience of $n=\alpha t+\beta$, for $\alpha>1,\beta>0$ without compromising the intersection property required for safety.
Should a quadratic solution for weak BA be developed, it could be used to improve the total resilience of our adaptive algorithm (instead of Momose and Ren's algorithm~\cite{momose2020optimal}).


\bibliography{references}

\newpage

\appendix

\section{Adaptive Byzantine Broadcast: Correctness}
\label{bb_correct}

\subsection{Correctness}

We start by proving the phase's requirements.
First, immediately from lines~\ref{l:return_valid}~--~\ref{l:return_bottom} we get that all correct processes return either $\bot$ or a valid $v$. Next, the following lemma shows that in non-silent phases with correct leaders all correct processes return a valid value.

\begin{lemma}
\label{lemma:only_valid}
If a phase is non-silent and its leader is correct, then all correct processes return a valid value.
\end{lemma}
\begin{proof}

If the leader is correct it broadcasts a $\helpreq$ message at line~\ref{l:help_bb_2}. All correct processes then answer at round 2. If the leader receives a value signed by the sender at line~\ref{l:r_s_v}, it broadcasts it at line~\ref{l:bc_v}. Otherwise, no correct processes received a value signed by the sender and sends an $\idk$ message at line~\ref{l:reply_idk}. Since $n=2t+1$, the leader receives at least $t+1$ $\idk$ messages (from the correct processes) and forms an $\idk$ certificate. It broadcasts this value at line~\ref{l:bc_idk}.
In both cases, all correct processes return a valid value at line~\ref{l:return_valid}.
\end{proof}

The next lemma proves that if all correct processes invoke a phase with a value other than $\bot$, then they can return only one type of a valid value -- a value signed by the sender.

\begin{lemma}
\label{lemma:one_type_valid}
If all correct processes invoke a phase with value $v\ne \bot$, there does not exists a value signed by $t+1$ processes in the system.
\end{lemma}
\begin{proof}
If all correct processes invoke a phase with value $v\ne \bot$, they reply to the $\helpreq$ messages at line~\ref{l:reply_v} and never send an $\idk$ message.
Since there are at most $t$ Byzantine processes, the leader cannot receive $t+1$ $\idk$ messages and form an $\idk$ certificate signed by $t+1$ different processes.

\end{proof}

We now prove the correctness of the BB algorithm. 
First, to be able to use the weak BA, all correct processes must execute it with valid initial values.

\begin{lemma}
All correct processes execute line~\ref{l:decision} with a valid initial value.
\end{lemma}
\begin{proof}
Let $p_i$ be a correct process.
In $n$ phases, there is one phase with $p_i$ as leader.
If $p_i$ has updated $v_i$ prior to that phase, it happened either line~\ref{l:rec_init_v} or at line~\ref{l:rec_val_v}.
Immediately from the code we get that in both cases $p_i$ updates a valid value.
If $p_i$ did not update a value, it initiates a non-silent phase, and by Lemma~\ref{lemma:only_valid} returns a valid value.
\end{proof}

Note that agreement and termination stem immediately from the code and the correctness of the weak BA. The following lemma proves validity.

\begin{lemma}
If \emph{sender} is correct, then all correct processes decide $v_{\emph{sender}}$.
\end{lemma}
\begin{proof}
If \emph{sender} is correct then all correct process learn $v_{\emph{sender}}$ by the end of round 1 and update their values at line~\ref{l:rec_init_v}.
By Lemma~\ref{lemma:one_type_valid}, in no phase can any process create a value signed by $t+1$ processes.
Thus, when executing the weak BA $v_{\emph{sender}}$ signed by the sender is the only valid value that exists in the run. By unique validity and since the \emph{sender} does not sign more than one initial value, $v_{\emph{sender}}$ is the only possible BA output.
It follows that all correct processes execute line~\ref{l:extract_v_s} and return the sender's value.

\end{proof}

We conclude the following theorem:
\begin{theorem}
Algorithm~\ref{alg:BB_main} solves BB.
\end{theorem}

\section{Weak BA: Correctness}
\label{weak_ba_correct}

We start by proving some lemmas about the phase's guarantees.
First, we prove that if the \emph{decision} is updated in a given phase, then its new value is valid.

\begin{lemma}
\label{lemma:validity}

If a correct process updates \emph{decision} during $invokePhase$, then $v$ is a valid decision value.

\end{lemma}
\begin{proof}

If a correct process updates its \emph{decision} value at line~\ref{l:set_decision} of $invokePhase$ then it must have received a finalized certificate signed by $\ceil[\big]{\frac{n+t+1}{2}}$ processes. Hence, at least one correct process $p'$ signed the decide message for $v$ at line~\ref{l:sign_dec}. 
By the code, $p'$ signed the decide message for $v$ if it received a commit certificate signed by $\ceil[\big]{\frac{n+t+1}{2}}$ processes. Hence, at least one correct process $p''$ signed the vote message for $v$ at line~\ref{l:sign_vote}. By the code, this is possible only if $v$ is a valid value (line~\ref{l:valid_v}).

\end{proof}

Next, we prove that all correct processes that update their \emph{decision} variable do so the same value. Moreover, at most one valid \emph{decide\_proof} can exist in the system. That is, a Byzantine process cannot devise a \emph{decide\_proof} that conflicts with any other \emph{decide\_proof} known by correct processes.

\begin{lemma}
\label{lemma:agreement_of_invoke}

All correct processes that update \emph{decision} during $invokePhase$ return the same \emph{decision}. In addition, at most one $\finalize$ certificate can be formed in all phases.

\end{lemma}
\begin{proof}


Assume that a correct process $p_i$ sets its decision value to $v$ in phase $l$ and another correct process $p_j$ sets its decision value to $w$ in phase $k\geq l$.

If $k=l$, then $p_i$ and $p_j$ set their decision value in the same round and they both receive a finalize certificate signed by $\ceil[\big]{\frac{n+t+1}{2}}$ different processes. At least one correct process signed both certificates and since correct processes sign at most one finalize message per phase, $v=w$.

For the case where $k>l$: in phase $l$, $p_i$ receives a finalize certificate signed by $\ceil[\big]{\frac{n+t+1}{2}}$ different processes. Thus, at least $\ceil[\big]{\frac{n+t+1}{2}}-t\geq \frac{n-t+1}{2}$ correct processes updated their commit to $v$ in that phase, along with a matching commit proof and $\emph{commit\_level}=l$ (line~\ref{line:rank}).
Since these processes are committed to $v$, they do not vote for any value proposed by a leader in the following phases.
Thus at most $n-t-\frac{n-t+1}{2}=\frac{n-t-1}{2}$ correct processes can sign a conflicting proposed value in any phase greater than $l$.
Since $\frac{n-t-1}{2}+t<\ceil[\big]{\frac{n+t+1}{2}}$, in any phase greater than $l$, no process can collect $\ceil[\big]{\frac{n+t+1}{2}}$ signatures on any value other than $v$.
Because processes that updated $\emph{commit\_level}=l$ do not accept commitments on values with commit level $d<l$ (line~\ref{line:checklevel}), at most $\frac{n-t-1}{2}$ can send a decide message on a value committed in phase $d<l$. Thus, at most $\frac{n-t-1}{2}+t<\ceil[\big]{\frac{n+t+1}{2}}$ decide messages for $w\neq v$ can be sent. Finally, no process can form and send a valid finalize certificate and decide upon any other value. Thus, $v=w$.

\end{proof}

We prove next that once a correct process is the leader of a non-silent phase, all correct processes return the same valid decision value by the end of that phase.

\begin{lemma}
\label{lemma:early_decision}
If a correct leader invokes $invokePhase$ in phase $k$ and $f<\frac{n-t-1}{2}$, then all correct processes return the same valid \emph{decision} by the end of the phase and this decision is a proposal of a correct process.

\end{lemma}
\begin{proof}

The leader broadcasts its value $v$ to all processes. If there is a correct process $p$ for which $commit\neq \bot$, it sends the message  $\lr{\commit,w,proof,j}_p$ to the leader. 
If the leader receives  $\lr{\commit,w,proof,j}_{p'}$ (from any process), it broadcasts in round 3 a commit certificate for $w$. Otherwise, since $f<\frac{n-t-1}{2}$, leader receives $\ceil[\big]{\frac{n+t+1}{2}}$ messages voting for $v$ and broadcasts a commit certificate for $v$.
Then, all correct processes send the leader a finalize messages on $v$ or $w$. Again, the leader receives $\ceil[\big]{\frac{n+t+1}{2}}$ messages finalizing $v$ and broadcasts a finalize certificate for $v$. Correct processes receive this message and update their \emph{decision} and \emph{decide\_proof} accordingly. Then, by the code they all return $v$.

\end{proof}

We now prove the correctness of the main BA algorithm.
The following two lemmas prove that although some processes may start executing $\mathcal{A}_\emph{fallback}$ at different times, they all successfully execute the fallback algorithm.

\begin{lemma}
\label{lemma:all_fb}
If some correct process executes the fallback algorithm in Algorithm~\ref{alg:BA_main}, all correct process do so and they all start at most $\delta$ time apart.

\end{lemma}
\begin{proof}
Let $p$ be the first correct process that executes the fallback algorithm at line~\ref{l:fb} of Algorithm~\ref{alg:BA_main} at time $t$.
This means that at time $t-2\delta$, $p$ broadcasts the fallback certificate to all other processes (line~\ref{l:bc_fb}). By synchrony, this certificate is guaranteed to arrive at all correct processes by $t-\delta$, causing them to execute the fallback algorithm by $t+\delta$ if they have not done so earlier.
\end{proof}

\begin{lemma}
\label{lemma:simulate}
Consider a synchronous algorithm $\mathcal{A}$.
Let $\sigma$ be a synchronized run of $\mathcal{A}$ defined as follows.
Let $t$ be the time that the first correct process starts executing $\mathcal{A}$ in $\sigma$. All correct processes start executing $\mathcal{A}$ by $t+\delta$. The round duration is $2\delta$.
In round $r$ that begins (locally) in $t_r$, round $r$ messages are processed if they are received in the time window $[t_r-\delta,t_r+2\delta]$.
Then $\sigma$ is a correct run of $\mathcal{A}$.

\end{lemma}
\begin{proof}

Consider a process $p$ that starts round $r$ at time $t^p_{r}$.
Let $p'$ be another correct process that starts round $r$ at time $t^{p'}_{r}$, and sends a message to $p$ in round $r$.
By assumption, $t^{p'}_{r}=t^p_{r}+\epsilon$ where $-\delta\le \epsilon\le \delta$, and a message sent by $p'$ at $t^{p'}_{r}$ arrives at time $t_a$ where $t^{p'}_{r}\le t_a \le t^{p'}_{r}+\delta$.
Note that round $r$ ends at $p$ at time $t^p_{r+1}=t^p_{r}+2\delta$.
Hence, $t^{p}_{r}-\delta \le t_a \le t^{p}_{r}+2\delta$, as needed.

\end{proof}

Next, we the following lemma states that if a correct process manages to reach a decision prior to the fallback algorithm, then this is the only possible decision. Moreover, this decision value must be a valid one.

\begin{lemma}
\label{lemma:possible_def}
If some correct process decides $v$ before executing the fallback algorithm, then all correct processes decide $v$ and $v$ is valid.
\end{lemma}

\begin{proof}

If there exists a correct process $p$ that decides at line~\ref{l:inv_phase}, then by Lemma~\ref{lemma:agreement_of_invoke} and the code all processes that decide at line~\ref{l:inv_phase} decide $v$ as well. Moreover, all other correct process that have not decided by line~\ref{l:end_phases}, send $\helpre$ messages. Process $p$ answers them and they all decide at line~\ref{l:got_help}.
Otherwise, no correct process decides at line~\ref{l:inv_phase} and they all send $\helpre$ messages at line~\ref{l:help}. Then, they all receive $t+1$ $\help$ messages and by the code perform the fallback algorithm. In addition, by the lemma assumption, it must be that $p$ decides $v$ at line~\ref{l:got_help}.

If correct processes execute the fallback algorithm, then by the code they all wait a time period of $2\delta$ before the execution, during which they receive all decisions made by other correct processes and update \emph{bu\_decision} accordingly (line~\ref{l:fb}). Specifically, they receive $v$ from $p$.
It follows from Lemma~\ref{lemma:agreement_of_invoke} that \emph{bu\_decision} is updated with the same value at all correct processes.
Thus, all correct processes execute $\mathcal{A}_{\emph{fallback}}$ with the same input, and by strong unanimity they set \emph{fallback\_val} to $v$ at line~\ref{l:fb}.

We now prove that $v$ is valid.
If $p$ decides $v$ at line~\ref{l:inv_phase}, then it must have updated \emph{decision} in the scope of the relevant phase. By Lemma~\ref{lemma:validity} this value is valid. Otherwise, if $p$ decides $v$ at line~\ref{l:got_help}, then the validity follows from the code. Hence, since $v$ is valid, all correct processes decide it by line~\ref{l:final_dec}.

\end{proof}

Finally, we are ready to prove the required BA properties.

\begin{lemma}[Agreement]
\label{lemma:agreement_ba}
In Algorithm~\ref{alg:BA_main} all correct process decide on the same value.
\end{lemma}
\begin{proof}

First, by Lemma~\ref{lemma:agreement_of_invoke}, all correct processes that decide in line~\ref{l:inv_phase} decide the same value $v$.
In addition, it follows from the same lemma that every correct process
that decides at line~\ref{l:got_help} after receiving a valid $\finalize$ certificate decides $v$, as at most one $\finalize$ certificate can be formed.

It is left to show that if not all correct processes decide before the fallback algorithm at line~\ref{l:fb}, they still decide upon the same value.
If at least one correct process $p$ receives a fallback certificate it follows from Lemma~\ref{lemma:all_fb} that all correct processes receive the certificate within at most $\delta$ time of $p$. 
Then, by the code, all correct process execute the fallback algorithm at line~\ref{l:fb} and by Lemma~\ref{lemma:simulate} and the fallback algorithm solves strong BA, providing agreement.
By Lemma~\ref{lemma:possible_def}, we get that processes that decide before running the fallback decide on the same value.

\end{proof}

\begin{lemma}[Termination]
\label{lemma:liveness}
In Algorithm~\ref{alg:BA_main} all correct process decide.
\end{lemma}
\begin{proof}

If not all correct processes decide before line~\ref{l:end_phases} and no correct process receives a fallback certificate, it follows that less than $t+1$ correct processes broadcast help messages at line~\ref{l:help}. Hence, at least one correct process $p$ has decided by line~\ref{l:end_phases}. Process $p$ receives all of the correct help messages at line~\ref{l:r_help} and answers them at line~\ref{l:a_help}.
All correct processes that asked for help then decide at line~\ref{l:got_help}.

It remains to examine the case that at least one correct process $p$ receives a fallback certificate. It follows from Lemma~\ref{lemma:all_fb} that all correct processes receive the certificate within at most $\delta$ time of $p$. 
Then, by the code, all correct process execute the fallback algorithm at line~\ref{l:fb} and by Lemma~\ref{lemma:simulate} and the fallback algorithm solves BA, providing termination.

\end{proof}

\begin{lemma}[Unique Validity]
\label{lemma:unique_val}
In Algorithm~\ref{alg:BA_main} if a correct process decides $v$ then either $v=\bot$ or $\textit{validate}(v) = \textit{true}$, and if $v=\bot$ then more than one valid value exists in the run.

\end{lemma}
\begin{proof}
Let $v$ be the decision value of a correct process in Algorithm~\ref{alg:BA_main}.
First, by lines~\ref{l:final_dec}~--~\ref{l:ret_def} $\textit{validate}(v) = \textit{true}$ or $v=\bot$.
We prove that if $v = \bot$, then at least two valid values exist in the run.

By the code, all processes execute the fallback algorithm with valid inputs (either their initial valid values, or a valid value they adopt at line~\ref{l:got_fb}). By strong unanimity of $\mathcal{A}_\emph{fallback}$, if all correct processes start with the same valid value $v'$, then $v'$ must be the returned decision value. This contradicts the fact that $\bot$ is returned at line~\ref{l:ret_def}. 
Therefore, not all correct processes execute $\mathcal{A}_\emph{fallback}$ with the same value. As they all execute the fallback algorithm with valid inputs, it follows that at least two valid values exist in the run.

\end{proof}

In addition, we need to prove that every correct process updates its \emph{decision} at most once.

\begin{lemma}
\label{lemma:dec_at_most_1}
In Algorithm~\ref{alg:BA_main} all correct process decide at most once.
\end{lemma}
\begin{proof}

Any correct process updates $\emph{decision}$ at line~\ref{l:inv_phase}, line~\ref{l:got_help} or lines~\ref{l:final_dec}~--~\ref{l:ret_def}.
In all cases, it only does so if $\emph{decision}=\emph{undecided}$. Since by the code it does not update $\emph{decision}$ to the value $\emph{undecided}$, it follows that $\emph{decision}$ is updated at most once.

\end{proof}

From Lemmas~\ref{lemma:agreement_ba}, \ref{lemma:liveness}, \ref{lemma:unique_val}, and~\ref{lemma:dec_at_most_1} we conclude:
\begin{theorem}
Algorithm~\ref{alg:BA_main} solves weak BA.
\end{theorem}

\section{Strong BA: Correctness}
\label{strong_ba_correct}

\begin{lemma}
\label{lem:all_fb}
If some correct process executes the fallback algorithm in Algorithm~\ref{alg:FF_strong_BA}, all correct process do so and they all start at at most $\delta$ time apart.
\end{lemma}

Proof is similar to Lemma~\ref{lemma:all_fb} in Section~\ref{section:aba}.

\begin{lemma}[Agreement]
\label{lemma:str_agreement}
In Algorithm~\ref{alg:FF_strong_BA} all correct process decide on the same value.
\end{lemma}
\begin{proof}

First, as correct processes only sign one $\decide$ message, every process that receives $QC_{\decide}(v)$ receives the same quorum certificate.
Thus, all correct processes that decide at line~\ref{ll:dec} decide the same $v$. If at least one correct process receives a fallback message then by Lemma~\ref{lem:all_fb}, they all execute the fallback algorithm at most $\delta$ time apart. Thus, if at least one correct process decides at line~\ref{ll:dec}, then all correct processes that have not yet decided learn about $v$ in the $2\delta$ safety window, and adopt it as their initial value for the fallback (line~\ref{ll:bu}).
It follows that all correct processes decide with the same input value $v$ and by strong unanimity this is the only possible decision.

\end{proof}

\begin{lemma}[Termination]
\label{lemma:str_live}
In Algorithm~\ref{alg:FF_strong_BA} all correct process decide.
\end{lemma}
\begin{proof}
If not all correct processes decide by line~\ref{ll:dec}, then a correct process broadcasts a $\fb$ message at line~\ref{ll:fall}. It follows from Lemma~\ref{lem:all_fb} that all correct processes receive the certificate within at most $\delta$ time of $p$. 
Then, by the code, all correct process execute the fallback algorithm at line~\ref{ll:fb} and by Lemma~\ref{lemma:simulate} and the fallback algorithm solves strong BA, providing termination.
\end{proof}

\begin{lemma}[Validity]
\label{lemma:str_val}
In Algorithm~\ref{alg:FF_strong_BA} if all correct processes propose the same value $v$, then the output is $v$.
\end{lemma}
\begin{proof}
Correct processes only send $\decide$ messages on values with valid $\propose$ quorum certificates. Note that such a quorum certificate can only be formed with $t+1$ unique signatures. Hence, if all correct processes propose the same value $v$, then the only possible $\propose$ quorum certificate is with $v$. As a result, the only possible $\decide$ quorum certificate is with $v$ as well.

The fallback algorithm is executed with either the original initial values or with a value that has a corresponding $\decide$ quorum certificate. Thus, if correct processes execute the fallback algorithm, they all start with $v$ and by strong unanimity of $\mathcal{A}_\emph{fallback}$, the decision is $v$.
\end{proof}

Finally, we prove that every correct process updates its \emph{decision} at most once.

\begin{lemma}
\label{lemma:dec_at_most_1_s}
In Algorithm~\ref{alg:FF_strong_BA} all correct process decide at most once.
\end{lemma}
\begin{proof}

Any correct process updates $\emph{decision}$ either at line~\ref{ll:dec} or at line~\ref{ll:dec_fb_v}.
In both cases, it only does so if $\emph{decision}=\bot$. Since it does not update $\emph{decision}$ to the value $\bot$ at any step of the algorithm, it follows that $\emph{decision}$ is updated at most once.

\end{proof}

From Lemmas~\ref{lemma:str_agreement}, \ref{lemma:str_live}, \ref{lemma:str_val}, and~\ref{lemma:dec_at_most_1_s} we conclude:

\begin{theorem}
Algorithm~\ref{alg:FF_strong_BA} solves binary strong BA.
\end{theorem}

\end{document}